\documentclass[11pt,a4paper]{article}
\pdfoutput=1
\usepackage[utf8]{inputenc}
\usepackage[T1]{fontenc}
\usepackage[english]{babel}
\usepackage{lmodern}
\usepackage{amsmath,amssymb,amsthm}
\usepackage[margin=2.6cm]{geometry}
\usepackage{booktabs}
\usepackage{microtype}
\usepackage{graphicx}

\newtheorem{theorem}{Theorem}
\newtheorem{proposition}{Proposition}
\newtheorem{corollary}{Corollary}
\theoremstyle{definition}

\newtheorem{assumption}{Assumption}
\newtheorem{remark}{Remark}

\title{\textbf{When Search Eats the Web\\ \large A Model of Corpus Erosion under Generative Extraction}}
\author{Sylvain Peyronnet\\ \small IBOU}
\date{}

\begin{document}
\maketitle

\begin{abstract}
\noindent Generative search engines (GSEs) answer user queries directly from crawled web content. The capture of value from the corpus without a visit returned to the source (we call this capture extraction) diverts the traffic that finances content production.
In response, publishers may restrict crawler access to their websites.

In this paper, we model the crawlable corpus as a common-pool resource: the crawlable commons.
It is described by three quantities: volume, average quality, and lifetime.
Under two types of responses of publishers we prove that extraction degrades all three at once: publishers opt out, renewal loses its funding, and content becomes more perishable. After a given erosion threshold, the corpus goes extinct. A myopic GSE can cross this threshold, a long-run oriented GSE stays below it.

We extend our model to several competing engines and prove, under a concavity condition on the steady-state value of the commons, that the symmetric equilibrium extraction rate is nondecreasing in their number and converges to the threshold. 
Adding users who strictly prefer direct answers, the assumption most favorable to extraction, we prove that the socially optimal extraction rate lies strictly below the erosion threshold, and no higher than the single engine's sustainable optimum. Finally, we discuss seven survival mechanisms.
\end{abstract}

\section{Introduction}\label{sec:intro}

There is an implicit contract between search engine and publishers: access to content is exchanged for traffic. A classical search engine crawls, indexes, and ranks webpages, then sends the user to the source (these webpages). The publisher then monetizes the visit. This contract is simple, and it has financed decades of online content production.

Generative search engines (GSEs) break this contract. A GSE reads documents, synthesizes a direct answer, and retains the user inside its interface. The traffic to the source disappears, and the publisher's revenue with it. 
Throughout this paper, {\em extraction} denotes this capture of value without a visit returned. The {\em extraction rate} $e$ is the fraction of a page's value the GSE keeps. Some studies have estimated this loss of traffic. For instance, \cite{khosravi2026} announces that Google's AI Overviews reduces the daily traffic of exposed Wikipedia articles by about 15\%.
The publishers are already reacting by blocking GSEs, withdrawing corpora, and litigating.

The current literature measures who loses what. In this paper we address a different question. The problem is what remains in the crawlable corpus, what stops entering it, and what happens to what is still there. We formalize the crawlable corpus as a commons in the economic sense \cite{hardin1968, ostrom1990} and we call it the crawlable commons. This is a resource whose value can be extracted by anyone, and whose maintenance is paid by nobody.
It is a textbook object of renewable resource economics \cite{gordon1954}: the corpus renews itself, depreciates, and is extracted from.

Our main contribution is to prove that a minimal model, in which extraction simultaneously affects the composition of the corpus, the funding of its renewal, and its lifetime, yields a critical regime governed by a single threshold. Moreover, we identify the parameters of the model that institutions can move.

The paper is organized as follows. Section~\ref{sec:related} is the related work. In section~\ref{sec:model} we introduce the crawlable commons model. Sections~\ref{sec:dilution}, \ref{sec:depletion} and~\ref{sec:decay} describe the three erosion channels (dilution, depletion, decay). Section~\ref{sec:critical} contains our core results: the erosion threshold, the simultaneous degradation below it, the condition for myopic collapse, the competition theorem, and the welfare result. The section~\ref{sec:mechanisms} presents the survival mechanisms that we propose, providing actionable frameworks for publishers and search engines.

\section{Related work}\label{sec:related}

We present here the scientific literature related to our work, categorized by the topic.

\medskip
\emph{\bf Traffic and substitution.} The effect of direct answers on source traffic has been estimated. For instance, \cite{khosravi2026} reports that Google's AI Overviews reduces the daily traffic of exposed Wikipedia articles by $\sim$ 15\%.
Blocking GSEs does not fully solve the problem. Large publishers that block AI crawlers lose traffic compared with not blocking \cite{strategic2026}.
These two facts match two ingredients of our model: extraction diverts traffic, and closing to crawlers has a cost.

\medskip
\emph{\bf Blocking and corpus composition.} Access restrictions are more common on the most popular sites: $\sim$ 25\% of the top thousand websites restrict AI crawlers, while it's $\sim$ 10\% of the top million \cite{bouchaud2025}. Restrictions also peak among media of high quality.
Opt-outs do not degrade GSEs general knowledge, but it does for specialized domains when major publishers opt out \cite{fan2025}.
These facts document separately the two effects at the center of our model: the best publishers are the first to block crawlers, and the remaining corpus degrades.

\medskip
\emph{\bf Attribution and restitution.} GSEs credit fewer sources than they read: several relevant sites remain uncited per query, and the number of citations returned per additional page read varies from one GSE to another \cite{strauss2026}.
To value each source's contribution to an answer, recent work applies Shapley values to the retrieved documents \cite{ye2025, nematov2025,maxshapley2025}.
This supports some of our survival mechanisms, but without addressing our question: what happens to the volume, quality and lifetime of the corpus when the traffic returned remain insufficient over the long run.

\medskip
\emph{\bf Economics of content production under intermediaries.} The question predates generative models: content aggregation changes how users read \cite{chiou2017}
and the quality newspapers choose to produce \cite{jeon2016}.
Recent work studies content creators facing generative AI. When AI can reuse what a creator publishes, the creator may stop investing effort. The paper \cite{ohayon2026} gives mechanisms to encourage creators to keep producing. \cite{acemoglu2026} shows that collective knowledge can collapse when AI removes the incentive to learn. 
However, a counterpoint exists. In a model of contributed public goods where AI takes the easy tasks, training on user contributions can raise overall quality \cite{gans2024}.
By contrast, in our model production is financed by traffic that extraction dries up, and quality publishers have alternatives that take them out of the commons.

\medskip
\emph{\bf Methods and building blocks.} Game theory has long been applied to information retrieval \cite{tennenholtz2019} and to the economics of search \cite{azzopardi2011}. GSEs build on retrieval-augmented generation (RAG) \cite{lewis2020}, and publishers adapt to them through generative engine optimization (GEO) \cite{aggarwal2024}.

\medskip
Each of these works addresses either traffic, blocking, attribution, or creator incentives. None couples them. Our contribution is this coupling: a single dynamic corpus described by three quantities (volume, quality, lifetime), degraded simultaneously, and subject to competition between several GSEs, with the resulting critical regime and rent dissipation.
\section{The crawlable commons model}\label{sec:model}

The object of our model is the crawlable commons: the part of the web that stays accessible to crawlers, treated as a common-pool resource.

\subsection{Agents}\label{subsec:agents}

The model has two types of agents.

\begin{itemize}
\item Publishers, indexed by their quality $\theta \in [0,1]$, distributed according to a continuous, strictly positive density $f$ on $[0,1]$, with cumulative distribution $F$ and mean $\mathbb{E}[\theta]$. The quality $\theta$ measures the value of the content produced by the publisher (for instance according to reliability, depth, originality, etc.).
\item A GSE, with an extraction rate $e \in [0,1]$. $e$ is the fraction of a page's value that it keeps without sending a visit back. This stands for several things at once: the fraction of visits the answer replaces, the fraction of value the engine keeps, and the pressure publishers respond to. The model does not separate them, $e$ can be seen as an aggregate of all this phenomenon. The case $e = 0$ corresponds to the classical search engine (all the value returns to the sources through visits), the case $e = 1$ to the full direct answer (the GSE keeps everything, no value returns to the publisher).
\end{itemize}

\subsection{State variables}\label{subsec:state}

The commons is described by two state variables in continuous time $t \in \mathbb{R}_+$:

\begin{itemize}
\item $N(t) > 0$: the volume of the crawlable corpus, measured in usable content, not in pages. A page counts for the value it still carries;
\item $Q(t) \in [0,1]$: the average quality of this corpus.
\end{itemize}

The informational value of the corpus is the product $S(t) = N(t)\, Q(t)$. The third quantity of the commons, its lifetime, is carried by the depreciation rate $\lambda$ defined in Section~\ref{sec:decay}. $N(t)$ and $Q(t)$ are corpus-level aggregates. The model does not track  individual publishers or pages. 

\subsection{The three channels}\label{subsec:channels}

Extraction takes visits away from publishers, who make their living by monetizing them. They can react in three ways, and our model treats each way as a channel.

The first channel is {\bf dilution}. A publisher may opt out, changing the composition of the crawlable corpus (Section~\ref{sec:dilution}).

The second channel is {\bf depletion}. A publisher that stays earns less, and lower revenue reduces the production of new content (Section~\ref{sec:depletion}).

The third channel is {\bf decay}. Durable pages are what the engine extracts best, so a publisher may shift toward perishable content, increasing the depreciation rate (Section~\ref{sec:decay}).

\section{First channel: dilution}\label{sec:dilution}

Dilution is the change in the composition of the corpus when publishers opt out. Each publisher compares what staying accessible to the crawler pays with what opting out pays. Theorem~\ref{thm:dilution} characterizes who stays.

\subsection{The publisher's participation choice}

At every instant, each publisher chooses to be open (accessible to the GSE's crawler) or closed (blocking the crawler by any technical means).

The revenue of an open publisher of quality $\theta$ is $r(1-e)\theta$, where $r$ is the revenue per unit of quality generated by organic traffic, and where $(1-e)$ reflects that the direct answer replaces the fraction $e$ of the visits, and voids the associated revenue.

The revenue of a closed publisher is $\omega\theta - \kappa$. Here $\omega$ is the revenue per unit of quality of the alternatives (subscription, licensing, direct audience). The recurring cost $\kappa$ of opting out is independent of $\theta$.

Revenue that does not depend on the participation choice (direct visits, newsletters, donations) enters both sides of the comparison and cancels out. The parameter $r$ carries the search-dependent part of the revenue. This component is confirmed empirically: large publishers that block GenAI crawlers lose traffic \cite{strategic2026}.

\begin{assumption}[parameters of the participation trade-off]\label{ass:participation}
$0 < \kappa < \omega < r$.
\end{assumption}

The assumption $\omega < r$ states that full organic traffic pays more than the alternatives. $\kappa > 0$ means that opting out is costly, while $\kappa < \omega$ states that opting out is profitable at least for the highest quality when traffic no longer pays.

A publisher of quality $\theta$ stays open if and only if $r(1-e)\theta \geq \omega\theta - \kappa$, that is $\theta\,[\omega - r(1-e)] \leq \kappa$. Indifferent publishers, with $\theta$ exactly at the cutoff, are considered to stay open.

\begin{theorem}[dilution by adverse selection]\label{thm:dilution}
Let $e_{\mathrm{out}} := 1 - (\omega-\kappa)/r$ be the opt-out threshold. Under Assumption~\ref{ass:participation}, $e_{\mathrm{out}} \in (0,1)$ and a publisher of quality $\theta$ stays open if and only if $\theta \leq \bar{\theta}(e)$, with
\begin{equation}
\bar{\theta}(e) \;=\;
\begin{cases}
1 & \text{if } e \leq e_{\mathrm{out}},\\[2pt]
\dfrac{\kappa}{\omega - r(1-e)} & \text{if } e > e_{\mathrm{out}},
\end{cases}
\qquad \text{where } \omega - r(1-e) \geq \kappa > 0 \text{ on the second branch}.
\end{equation}
The function $\bar{\theta}$ is continuous, equal to $1$ on $[0, e_{\mathrm{out}}]$, strictly decreasing on $[e_{\mathrm{out}}, 1]$, with $\bar{\theta}(1) = \kappa/\omega < 1$. Consequently:
\begin{itemize}
\item[(a)] the mass of open publishers $F(\bar{\theta}(e))$ equals $1$ on $[0,e_{\mathrm{out}}]$ then decreases strictly;
\item[(b)] the average quality of open publishers $\bar{Q}(e) := \mathbb{E}[\theta \mid \theta \leq \bar{\theta}(e)]$ equals $\mathbb{E}[\theta]$ on $[0,e_{\mathrm{out}}]$ then decreases strictly;
\item[(c)] for any $e > e_{\mathrm{out}}$, the closed publishers are exactly those of quality above $\bar{\theta}(e)$.
\end{itemize}
\end{theorem}

\begin{proof}
If $\omega \leq r(1-e)$, the participation condition holds for every $\theta$ since $\kappa > 0$. If $\omega > r(1-e)$, it rewrites as $\theta \leq \kappa/(\omega - r(1-e))$, which is at most $1$ if and only if $\omega - r(1-e) \geq \kappa$, that is $e \geq e_{\mathrm{out}}$. Assumption~\ref{ass:participation} gives $0 < (\omega-\kappa)/r < 1$, so $e_{\mathrm{out}} \in (0,1)$.

On $(e_{\mathrm{out}}, 1]$, $\omega - r(1-e)$ increases strictly with $e$, so $\bar{\theta}$ decreases strictly. Continuity at $e_{\mathrm{out}}$ holds since $\bar{\theta}(e_{\mathrm{out}}) = \kappa/\kappa = 1$.

The mass of open publishers and the average quality inherit these properties. $F$ is strictly increasing since $f > 0$. So is $s \mapsto \mathbb{E}[\theta \mid \theta \leq s]$: for $0 < s < s'$, $\mathbb{E}[\theta \mid \theta \leq s']$ is a strict convex combination of $\mathbb{E}[\theta \mid \theta \leq s]$ and $\mathbb{E}[\theta \mid s < \theta \leq s']$, both events have positive mass since $f > 0$, and the second term exceeds $s$.
\end{proof}

\begin{remark}\label{rem:plateau}
One could expect opting out to begin at $e_0 = 1 - \omega/r$, the rate at which the alternatives start paying more than organic traffic per unit of quality. It does not, because opting out also incurs the fixed cost $\kappa$. Between $e_0$ and $e_{\mathrm{out}}$, the advantage of opting out is real but too small to repay $\kappa$. We call the interval $(e_0, e_{\mathrm{out}})$ the dead zone. Its width $e_{\mathrm{out}} - e_0 = \kappa/r$ measures the cost of opting out. The model therefore predicts a plateau: below $e_{\mathrm{out}}$ no publisher opts out and beyond it the best publishers leave first (Theorem~\ref{thm:dilution} (c)).
\end{remark}

\subsection{Interpretation}

The best publishers leave first. This is adverse selection, but not through Akerlof’s classical information mechanism \cite{akerlof1970}. Quality is observable here, and high-quality publishers leave because the alternatives pay more as $\theta$ increases, while the opt-out cost does not.
High-quality publishers opt out or negotiate direct licenses. The crawlable corpus consequently becomes negatively selected (low-quality content, content farms, etc. stay in the corpus). 

The consequence reverses the usual framing of blocking. It filters the crawlable corpus, and the open web becomes an adverserly selected subset of the web rather than a sample.
Measurements already show this pattern \cite{bouchaud2025,fan2025}.

\section{Second channel: depletion}\label{sec:depletion}

The second channel is the renewal of the corpus. New content is paid for by the publishers who remain open, so the flow of new content depends on their revenue.

\begin{assumption}[funded renewal and capacity]\label{ass:renewal}
The renewal flow of new content is proportional to the aggregate revenue of open publishers and to the existing volume, and saturates as the volume approaches a capacity $K$:
\begin{equation}
P(t) \;=\; n(e)\, N(t) \left(1 - \frac{N(t)}{K}\right),
\qquad
n(e) \;:=\; \pi\, r\,(1-e)\, \mu(e),
\qquad
\mu(e) := \int_0^{\bar{\theta}(e)} \theta f(\theta)\, d\theta,
\end{equation}
where $\pi > 0$ is the productivity of the publishers investment. It converts revenue into content. New content has the average quality of open publishers, $\bar{Q}(e)$.
\end{assumption}

An open publisher of quality $\theta$ earns $r(1-e)\theta$, so the aggregate revenue of open publishers is $r(1-e)\,\mu(e)$. The factor $\mu(e) = F(\bar{\theta}(e))\,\bar{Q}(e)$ is their mass times their average quality. $n(e)$ is this revenue multiplied by productivity $\pi$. Extraction thus reduces renewal directly through $(1-e)$, and indirectly through $\mu(e)$, which shrinks when the best publishers opt out (Theorem~\ref{thm:dilution}).

The factor $N(1-N/K)$ is the standard logistic form for renewable resource. The capacity $K$ captures limits not modeled explicitely, such as reader attention, supply of topics, etc.

\section{Third channel: decay}\label{sec:decay}

We now formalize the third channel, that is the lifetime of content. All information expires, it can be in a few hours for news articles to years for encyclopedic content. Let $\lambda$ denote the depreciation rate of the corpus: the fraction of value lost per unit of time. Unlike in models of physical phenomena, our $\lambda$ is partly dependent on, and reflects, publishers choices. 

Higher extraction gives publishers an incentive to shift value from durable content toward more perishable formats, such as live data, real-time feeds, interactive tools, or social (UGC) content.

\begin{assumption}[shift toward perishable content]\label{ass:perishability}

The depreciation rate increases linearly with the extraction rate:
\begin{equation}
\lambda(e) = \lambda(0) + \delta\, e, \qquad \lambda(0) > 0,\; \delta > 0,\; \lambda(1) < 1,
\end{equation}
where $\lambda(0)$ is the depreciation rate without extraction and $\delta$ its sensitivity to extraction. $\lambda(1) < 1$ ensures that $\lambda$ remains interpretable as a fraction of value lost per unit of time.
\end{assumption}

Rather than explicitly modeling a publisher's format optimization, we use a reduced-form linear specification to capture the behavioral shift toward perishable content as extraction increases.
The model applies the resulting depreciation rate $\lambda(e)$ uniformly to new and already existing content.

\section{The critical regime}\label{sec:critical}

Now that we have described the three channels, we assemble them and prove the core results: the erosion threshold, the simultaneous degradation regime, and the outcomes of myopic and competitive extraction.

\subsection{Full dynamics}\label{subsec:dynamics}

The three channels combine into a two-equation system:
\begin{align}
\dot{N}(t) &= n(e)\, N(t)\!\left(1 - \frac{N(t)}{K}\right) - \lambda(e)\, N(t), \\[2pt]
\dot{Q}(t) &= n(e)\!\left(1 - \frac{N(t)}{K}\right)\!\bigl(\bar{Q}(e) - Q(t)\bigr).
\end{align}

At a fixed extraction rate $e$, the set of open publishers is constant over time (Theorem~\ref{thm:dilution}), so the corpus evolves only through renewal and depreciation. 
The volume equation has the standard form of renewable-resource dynamics, with the difference that extraction does not remove content directly. It acts here through both coefficients, reducing renewal via $n(e)$ and accelerating decay via $\lambda(e)$.

Depreciation drops out of the average-quality equation because it affects all quality levels at the same rate. Only new content moves the average: the corpus progressively inherits the quality $\bar{Q}(e)$ of the publishers that remain open.

\subsection{The threshold}\label{subsec:threshold}

A single threshold on the extraction rate separates a corpus that survives from a corpus that goes extinct. Proposition~\ref{prop:threshold} identifies this threshold.

\begin{proposition}[erosion threshold]\label{prop:threshold}

Assume $n(0) > \lambda(0)$. Then there exists a unique $e^* \in (0,1)$ such that $n(e^*) = \lambda(e^*)$. For any initial state $N(0) \in (0, K)$, there are three regimes:
\begin{itemize}
\item if $e < e^*$, the volume converges to the interior steady state
$N^*(e) = K\bigl(1 - \lambda(e)/n(e)\bigr) > 0$ and the corpus quality converges to $\bar{Q}(e)$;
\item if $e = e^*$, the volume vanishes critically, in a hyperbolic regime: $N(t) \sim K/(n(e^*)\, t)$;
\item if $e > e^*$, the volume vanishes exponentially, at asymptotic rate $\lambda(e) - n(e)$.
\end{itemize}
\end{proposition}

\begin{proof}
We first show that $n - \lambda$ is continuous and strictly decreasing. By Theorem~\ref{thm:dilution} and Assumption~\ref{ass:renewal}, $(1-e)$ decreases strictly and $\mu(e)$ decreases, so $n$ decreases strictly. Assumption~\ref{ass:perishability} ensures that $\lambda$ increases strictly. Both are continuous, so is $n - \lambda$. Now, $n(0) > \lambda(0)$ by assumption, while $n(1) = 0 < \lambda(1)$ by construction. Using the intermediate value theorem we have the existence of $e^*$, its uniqueness coming from the strict monotonicity.

We then prove the three regimes in turn. At fixed $e$, the equation in $N$ is a logistic with linear removal. For $n > \lambda$, it rewrites as $\dot{N} = (n-\lambda) N (1 - N/N^*)$ with $N^* = K(1 - \lambda/n)$, whose interior steady state attracts every $N(0) > 0$. For $n = \lambda$, the equation becomes $\dot{N} = -(n/K)N^2$, whose explicit solution $N(t) = N(0)/(1 + (n/K)N(0) t)$ gives the equivalent $K/(n t)$. For $n < \lambda$, we have $\dot{N} \leq -(\lambda - n) N$, so $N(t) \to 0$, and the asymptotic rate is $\lambda - n$.

Finally, the convergence of $Q(t)$ to $\bar{Q}(e)$ follows from the second equation, since its factor $n(e)(1 - N(t)/K)$ remains strictly positive when $N^* < K$.
\end{proof}

Below the threshold, the crawlable corpus survives at a smaller volume $N^*(e)$, with the quality $\bar{Q}(e)$ of the publishers still open. It's a degraded, but stable, web. At the threshold and beyond, the volume goes to zero, slowly at $e^*$ itself, exponentially fast past it. The equality $n(e^*) = \lambda(e^*)$ says that funded renewal covers depreciation, beyond $e^*$ depreciation is greater than renewal.

One might expect the three channels to produce three distinct critical points. They don't, instead they feed the same two equations, so a single threshold governs all three.

\subsection{The surviving corpus}\label{subsec:degradation}

Once opting out has begun, the three quantities of the surviving corpus (volume, quality, lifetime) degrade together as extraction rises.

\begin{proposition}[simultaneous degradation]\label{prop:degradation}
On $(e_{\mathrm{out}}, e^*)$, the three steady-state quantities of the commons degrade strictly as $e$ increases: the volume $N^*(e)$ decreases, the quality $\bar{Q}(e)$ decreases, the content half-life $\ln 2/\lambda(e)$ decreases. 
\end{proposition}

\begin{proof}
For the volume, we write $N^*(e) = K\bigl(1 - \lambda(e)/n(e)\bigr)$. Since $\lambda$ increases strictly and $n$ decreases strictly, the quotient $\lambda/n$ increases strictly, so $N^*$ decreases strictly. The decrease of $\bar{Q}$ on $(e_{\mathrm{out}}, 1]$ is Theorem~\ref{thm:dilution}~(b). The decrease of the half-life follows from Assumption~\ref{ass:perishability}, since $\delta > 0$.
\end{proof}

Each quantity degrades for its own reason. The volume $N^*(e)$ falls because renewal lacks revenue while depreciation accelerates. The quality $\bar{Q}(e)$ falls with the best publishers opting out first (Theorem~\ref{thm:dilution}~(b)). The half-life $\ln 2/\lambda(e)$ shortens as production shifts toward perishable content (Assumption~\ref{ass:perishability}). The three channels also reinforce one another: a higher $e$ pushes more publishers out, which reduces $\mu(e)$ and slows renewal further, while depreciation keeps accelerating.

Between the two thresholds, degradation is continuous and reversible: lowering the extraction rate improves the steady state.

\subsection{Myopia and the sustainable optimum}\label{subsec:myopia}

So far the extraction rate has been a parameter. We now compare two possible choices by the GSE. One is a myopic rate that maximizes answer quality at every instant, the other is a long-run rate that accounts for its impact on the corpus.

The GSE instantaneous payoff is:
\begin{equation}
u(e, t) \;=\; e\; \bar{Q}(e)\; N(t).
\end{equation}

We restrict our attention to stationary extraction rates. The myopic engine maximizes $u(e,t)$ at every instant; the current volume $N(t)$ enters as a fixed positive factor, so this amounts to maximizing $e\,\bar{Q}(e)$, whatever the volume. 
We denote by $e_{\mathrm{myo}}$ a rate that achieves this maximum and call it the myopic rate. If several rates achieve it, the results below hold for each of them. The long-run oriented engine maximizes the average payoff, the answer quality delivered per unit of time in the long run: $J(e) := \lim_{T \to \infty} \frac{1}{T}\int_0^T u(e,t)\, dt$, which equals $e\,\bar{Q}(e)\,N^*(e)$ below the erosion threshold and $0$ on $[e^*, 1]$ (Proposition~\ref{prop:threshold}). 
Two criteria can evaluate an infinite horizon: discounting future payoffs, or averaging them per unit of time. Both are finite here, since the volume is bounded by $K$. We choose the average payoff because it depends on the steady state alone, which gives $J$ its closed form. A discounted criterion would also value the transition.

\begin{theorem}[myopia and the sustainable optimum]\label{thm:myopia}
\leavevmode
\begin{itemize}
\item[(a)] Every myopic rate satisfies $e_{\mathrm{myo}} \geq e_{\mathrm{out}}$: no rate below the opt-out threshold is myopically optimal.
\item[(b)] The average payoff $J$ is continuous on $[0,1]$, zero at $0$ and on $[e^*, 1]$, strictly positive on $(0, e^*)$: it attains its maximum at an interior rate $e_{\mathrm{eng}} \in (0, e^*)$, the engine's sustainable optimum.
\item[(c)] If $e^* < e_{\mathrm{out}}$, then $e_{\mathrm{myo}} > e^*$: by Proposition~\ref{prop:threshold}, $u(e_{\mathrm{myo}}, t) \to 0$ even though the engine maximizes $u$ at every instant. The condition $e^* < e_{\mathrm{out}}$ reads
\begin{equation}
\pi\,(\omega - \kappa)\, \mathbb{E}[\theta] \;<\; \lambda(0) + \delta \left(1 - \frac{\omega-\kappa}{r}\right).
\end{equation}
\end{itemize}
\end{theorem}

\begin{proof}
(a) On $[0,e_{\mathrm{out}}]$, $\bar{Q}(e) = \mathbb{E}[\theta]$ (Theorem~\ref{thm:dilution}~(b)), so $e\,\bar{Q}(e)$ is strictly increasing. Every maximizer therefore satisfies $e_{\mathrm{myo}} \geq e_{\mathrm{out}}$.

(b) Below the threshold, $J(e) = e\,\bar{Q}(e)\,N^*(e)$, and $J = 0$ on $[e^*, 1]$. Continuity on $[0,1]$ holds because $N^*(e) \to 0$ as $e \to e^{*-}$. On $(0, e^*)$, $J$ is strictly positive, since $\bar{Q} > 0$ and $N^* > 0$, while $J(0) = 0$. A continuous function on a compact set, zero at the boundary and positive inside, attains its maximum inside, so we have $e_{\mathrm{eng}} \in (0, e^*)$.

(c) If $e^* < e_{\mathrm{out}}$, then $e_{\mathrm{myo}} \geq e_{\mathrm{out}} > e^*$ by (a), and Proposition~\ref{prop:threshold} gives $N(t) \to 0$, so we have $u(e_{\mathrm{myo}}, t) \to 0$. For the condition, note that $e^* < e_{\mathrm{out}}$ is equivalent to $n(e_{\mathrm{out}}) < \lambda(e_{\mathrm{out}})$, by strict monotonicity of $n - \lambda$. On $[0, e_{\mathrm{out}}]$, $\mu = \mathbb{E}[\theta]$, and $1 - e_{\mathrm{out}} = (\omega-\kappa)/r$, so $n(e_{\mathrm{out}}) = \pi\, r\,(1-e_{\mathrm{out}})\,\mathbb{E}[\theta] = \pi(\omega-\kappa)\,\mathbb{E}[\theta]$. Writing $n(e_{\mathrm{out}}) < \lambda(e_{\mathrm{out}})$ with these values gives the stated inequality.
\end{proof}

The myopic objective contains no brake below the opt-out threshold. Quality is flat (Theorem~\ref{thm:dilution}~(b)) so raising $e$ improves the answers at any given volume and pushes the myopic engine at least to the point where the best publishers start opting out. 

The long-run trade-off is different. Raising $e$ increases the extracted fraction but lowers the steady-state value of the commons, $\bar{Q}(e)\,N^*(e)$, which disappears at the erosion threshold. Maximizing the average payoff is therefore enough to keep a single long-run oriented engine below $e^*$. Under the condition of part (c), myopia takes the engine past the threshold: maximizing answer quality at every instant eventually drives that quality to zero.

\begin{remark}\label{rem:condition-c}
$e^* < e_{\mathrm{out}}$ states that erosion bites before the first opt-out: the revenue that survives at the opt-out margin, $\pi(\omega-\kappa)\,\mathbb{E}[\theta]$, is insufficient to fund a renewal covering depreciation. In this regime, the myopic engine crosses the threshold. In the complementary regime ($e^* \geq e_{\mathrm{out}}$), the myopic rate can lie on either side of the threshold depending on the shape of $\bar{Q}$ beyond $e_{\mathrm{out}}$: myopic collapse is then possible, not necessary. 
\end{remark}

\subsection{Competition between engines}\label{subsec:competition}

With several (say $m$) engines extracting from the same commons, the decisive argument is not the myopia of one engine: even long-run oriented engines overextract because each engine bears the full cost of its own restraint but collects only one $m$-th of the value that restraint preserves.

The $m$-engine game is as follows: users queries are split equally among $m \geq 1$ engines, each engine chooses a stationary extraction rate $e_i \in [0,1]$, and the commons responds to the average rate $\bar{E} = \frac{1}{m}\sum_j e_j$ (the revenue of an open publisher is $r(1-\bar{E})\theta$, all functions of the previous sections apply at $\bar{E}$). 

Let $\Phi(x) := \bar{Q}(x)\, N^*(x)$ denote the steady-state value of the commons, extended by $0$ on $[e^*, 1]$: $\Phi$ is continuous, strictly decreasing on $[0, e^*]$, with $\Phi(0) > 0$ and $\Phi(e^*) = 0$. Engine $i$ serves one $m$-th of the queries at its own rate $e_i$, so its average payoff is
\begin{equation}
J_i(e_i, e_{-i}) \;=\; \frac{1}{m}\; e_i\; \Phi(\bar{E}).
\end{equation}

For $m = 1$, we recover $J$ from the previous subsection and its maximizer $e_{\mathrm{eng}}$. Summed over the $m$ engines, the average payoffs add up to $\bar{E}\,\Phi(\bar{E})$: the aggregate payoff.

By Theorem~\ref{thm:dilution}, $\bar{\theta}$ is constant below $e_{\mathrm{out}}$ and strictly decreasing beyond, so $\bar{Q}$, $N^*$ and $\Phi$ are differentiable except at $e_{\mathrm{out}}$, where their slope jumps. This point is the kink.

\begin{assumption}[regularity]\label{ass:regularity}
$\Phi$ is concave on $[0, e^*]$, differentiable except possibly at the kink $e_{\mathrm{out}}$, where its slope decreases.
\end{assumption}

Assumption~\ref{ass:regularity} is not implied by Assumptions~\ref{ass:participation} to~\ref{ass:perishability}: it imposes that each additional unit of extraction destroys at least as much steady-state value as the previous one, and $\Phi$ combines the quotient $\lambda/n$ inside $N^*$, the average quality of open publishers $\bar{Q}$ (a truncated mean), and the kink at $e_{\mathrm{out}}$, which do not guarantee this in general. The assumption is satisfiable (see Appendix~\ref{sec:appendix}). 

\begin{theorem}[competition and rent dissipation]\label{thm:competition}
Under Assumptions~\ref{ass:participation} to~\ref{ass:regularity} and the conditions of Proposition~\ref{prop:threshold}:
\begin{itemize}
\item[(a)] there exists a unique sustainable symmetric Nash equilibrium (with average rate strictly below $e^*$): all engines play the same rate $\hat{e}(m) \in (0, e^*)$, characterized by the subgradient condition
\begin{equation}
m\,\Phi(\hat{e}) + \hat{e}\,\Phi'_-(\hat{e}) \;\geq\; 0 \;\geq\; m\,\Phi(\hat{e}) + \hat{e}\,\Phi'_+(\hat{e}),
\end{equation}
where $\Phi'_-$ and $\Phi'_+$ denote the left and right derivatives, which reduces to $m\,\Phi(\hat{e}) + \hat{e}\,\Phi'(\hat{e}) = 0$ at every point of differentiability;
\item[(b)] $\hat{e}(1) = e_{\mathrm{eng}}$; $\hat{e}(m)$ is nondecreasing in $m$, strictly increasing whenever $\hat{e}(m) \neq e_{\mathrm{out}}$, and the possible kink plateau covers at most finitely many values of $m$; in particular $\hat{e}(m) > e_{\mathrm{eng}}$ for every $m \geq 2$, except possibly when $e_{\mathrm{eng}} = \hat{e}(m) = e_{\mathrm{out}}$;
\item[(c)] $\hat{e}(m) \to e^*$ as $m \to \infty$, and the aggregate payoff $\hat{e}(m)\,\Phi(\hat{e}(m)) \to 0$: the rent of the commons, what its exploitation still earns in total, dissipates entirely;
\item[(d)] every sustainable equilibrium whose average rate differs from $e_{\mathrm{out}}$ is symmetric, hence equal to $\hat{e}(m)$; at average rate $e_{\mathrm{out}}$, asymmetric equilibria may exist.
\end{itemize}
\end{theorem}

\begin{proof}
At a symmetric profile $e_i = e$, the derivative of an engine's payoff w.r.t. its own rate equals $G_m(e)/m^2$, where $G_m(e) := m\,\Phi(e) + e\,\Phi'(e)$, written $G_m^{\pm}$ when the one-sided derivatives $\Phi'_{\pm}$ are needed at the kink.

(a) We first prove that engine $i$'s payoff is concave in its own rate. At every point of differentiability, its second derivative with respect to $e_i$ equals $\frac{2}{m^2}\Phi'(\bar{E}) + \frac{e_i}{m^3}\Phi''(\bar{E})$, which is nonpositive because $\Phi' \leq 0$ and $\Phi'' \leq 0$. At the kink, the slope of $\Phi$ drops, and a drop in slope preserves concavity. By concavity, a rate $e_i$ is a best response to the other engines' rates if and only if it satisfies the subgradient condition, which at a symmetric profile reads $G_m^-(e) \geq 0 \geq G_m^+(e)$. Wherever $\Phi$ is differentiable, $G_m^-$ and $G_m^+$ coincide and decrease strictly, because $(m+1)\Phi' + e\Phi'' < 0$. At the kink, $G_m^+(e_{\mathrm{out}}) < G_m^-(e_{\mathrm{out}})$, again because the slope of $\Phi$ drops. Since $G_m$ decreases, the condition can hold at only one point. It does hold in $(0, e^*)$, because $G_m^{\pm}(0) = m\,\Phi(0) > 0$ while $G_m^{\pm}(e^{*-}) = e^*\,\Phi'_{\pm}(e^{*-}) < 0$. This unique point is a zero of $G_m$, or $e_{\mathrm{out}}$ itself when $G_m^-(e_{\mathrm{out}}) \geq 0 > G_m^+(e_{\mathrm{out}})$.

(b) For $m = 1$, the average payoff equals $J$, so the sustainable equilibria of the one-engine game are the maximizers of $J$, which lie in $(0, e^*)$ by Theorem~\ref{thm:myopia}~(b). By the uniqueness proved in (a), $e_{\mathrm{eng}}$ is unique and $\hat{e}(1) = e_{\mathrm{eng}}$. For $m' > m$, we have $G_{m'} = G_m + (m'-m)\,\Phi > G_m$ wherever $\Phi > 0$. Adding engines thus raises $G$, so its zero moves right and $\hat{e}(m)$ is nondecreasing. Away from the kink the growth is strict, because $\hat{e}(m)$ is then a zero of $G_m$, so $G_{m+1}(\hat{e}(m)) = \Phi(\hat{e}(m)) > 0$ and $\hat{e}(m+1) > \hat{e}(m)$. The kink can hold $\hat{e}$ only for finitely many $m$, because $G_m^+(e_{\mathrm{out}})$ grows by $\Phi(e_{\mathrm{out}}) > 0$ with each added engine and eventually turns positive, after which the subgradient condition fails at $e_{\mathrm{out}}$. Finally, if $e_{\mathrm{eng}} \neq e_{\mathrm{out}}$, then $G_1^+(e_{\mathrm{eng}}) = 0$ at this interior maximizer, so $G_m^+(e_{\mathrm{eng}}) = (m-1)\,\Phi(e_{\mathrm{eng}}) > 0$ for $m \geq 2$, hence $\hat{e}(m) > e_{\mathrm{eng}}$. The case $e_{\mathrm{eng}} = e_{\mathrm{out}}$ is the stated exception.

(c) For any $e < e^*$, $G_m^+(e) \to +\infty$ as $m$ grows, so $\hat{e}(m)$ exceeds $e$ for $m$ large enough. Since $\hat{e}(m)$ is nondecreasing and exceeds every $e < e^*$ eventually, it converges to $e^*$. By continuity of $\Phi$, the aggregate payoff $\hat{e}(m)\,\Phi(\hat{e}(m))$ tends to $e^*\,\Phi(e^*) = 0$.

(d) Consider a sustainable equilibrium with average rate $\bar{E} < e^*$, and suppose first $\bar{E} \neq e_{\mathrm{out}}$, so that $\Phi'(\bar{E})$ exists. Concavity and strict decrease of $\Phi$ make every one-sided derivative used here finite and strictly negative. No engine plays $0$, because its payoff derivative there equals $\frac{1}{m}\Phi(\bar{E}) > 0$. If some but not all engines played $1$, the others would be interior with $e_j = -m\Phi(\bar{E})/\Phi'(\bar{E}) < 1$, while the corner $e_i = 1$ requires $-m\Phi(\bar{E})/\Phi'(\bar{E}) \geq 1$, a contradiction. All engines at $1$ would force $\bar{E} = 1 \geq e^*$, excluded.\\ Every rate is therefore interior and satisfies $m\,\Phi(\bar{E}) + e_i\,\Phi'(\bar{E}) = 0$, so the $e_i$ are all equal and the equilibrium is $\hat{e}(m)$. Suppose now $\bar{E} = e_{\mathrm{out}}$. An interior rate must lie in $\bigl[-m\Phi/\Phi'_+,\, -m\Phi/\Phi'_-\bigr]$, everything evaluated at $e_{\mathrm{out}}$, and the corner rate $1$ only requires $m\Phi + \Phi'_- \geq 0$ there. Some engine lies in $(0,1)$ because $e_{\mathrm{out}} < 1$, which puts $-m\Phi/\Phi'_+$ below $1$, so every rate, corner included, lies in this interval intersected with $(0,1]$. Conversely, by concavity of each engine's payoff, every profile with average rate $e_{\mathrm{out}}$ and components in this set is an equilibrium. These are the asymmetric equilibria of the statement.
\end{proof}

Competition has a unique sustainable symmetric equilibrium for every number of engines. As the number increases, the equilibrium extraction rate rises from the single-engine sustainable optimum toward the erosion threshold, with the kink possibly holding it constant for finitely many values of $m$. The commons survives at every $m$, but its aggregate payoff vanishes as $m$ grows. Away from the kink, every sustainable equilibrium is symmetric: the sequence $\hat{e}(m)$ describes the entire equilibrium set.

\begin{remark}[mutual-ruin equilibria]\label{rem:mutual-ruin}
Theorem~\ref{thm:competition} characterizes the equilibria where the commons lives. For $m \geq 1/(1-e^*)$, dead equilibria exist as well: with all engines at $1$, everyone earns zero, and a lone deviation to $0$ lowers the average rate only to $(m-1)/m \geq e^*$, so no engine can revive the commons alone. These mutual-ruin equilibria are Pareto-dominated: every engine does strictly better at $\hat{e}(m)$. During the decline, the engines still extract from the remaining volume, and these gains count for nothing in an average over an infinite horizon: this is why no engine gains by changing its rate alone.
\end{remark}

Long-run oriented competition therefore does not cross the threshold. The rate $\hat{e}(m)$ climbs toward $e^*$, pausing at the kink for at most finitely many $m$. The commons survives at every $m$, with a volume $N^*(\hat{e}(m))$ and an aggregate payoff that both tend to zero as $m$ grows.

This is the rent dissipation result for open-access resources, established for fisheries \cite{gordon1954} and transposed to the crawlable corpus. It is the tragedy of the commons \cite{hardin1968} in the strict sense: the gain of extraction stays with the engine while the damages are shared, so each engine internalizes one $m$-th of the future of the commons.
The myopic case remains: Proposition~\ref{prop:myopic-comp} shows that with enough engines, myopic competition pushes extraction to its maximum.

\begin{proposition}[myopic competition]\label{prop:myopic-comp}
In the instantaneous game where each engine maximizes $\frac{1}{m}\,e_i\,\bar{Q}(\bar{E})\,N(t)$:
\begin{itemize}
\item[(a)] every symmetric equilibrium satisfies $\hat{e}_{\mathrm{myo}}(m) \geq e_{\mathrm{out}}$;
\item[(b)] let $\|\bar{Q}'\|_\infty$ be the supremum over $[0,1]$ of the absolute values of the derivatives of $\bar{Q}$ where they exist and of its one-sided derivatives at the kink, and $\bar{m} := \|\bar{Q}'\|_\infty\,/\,\bar{Q}(1)$; for every integer $m > \bar{m}$, playing $e_i = 1$ is strictly dominant: the unique Nash equilibrium is $\bar{E} = 1 > e^*$, and the extinction regime of Proposition~\ref{prop:threshold} applies.
\end{itemize}
\end{proposition}

\begin{proof}
(a) Fix a symmetric profile at a rate $e < e_{\mathrm{out}}$. On a neighborhood of $e$, $\bar{Q}$ is constant (Theorem~\ref{thm:dilution}~(b)), so the derivative of engine $i$'s payoff with respect to its own rate equals $\frac{1}{m}N(t)\,\bar{Q}(e)$, which is strictly positive. Deviating upward is therefore profitable, and no symmetric equilibrium exists below $e_{\mathrm{out}}$.

(b) The function $\bar{Q}$ is bounded below by $\bar{Q}(1) = \mathbb{E}[\theta \mid \theta \leq \kappa/\omega]$, which is strictly positive, and its derivative is bounded on $[0,1]$, because the one-sided slopes of $\bar{\theta}$ are finite, including at $e_{\mathrm{out}}$. The derivative of engine $i$'s payoff with respect to its own rate satisfies
\begin{equation*}
\frac{1}{m}N(t)\Bigl[\bar{Q}(\bar{E}) + \frac{e_i}{m}\bar{Q}'(\bar{E})\Bigr] \;\geq\; \frac{1}{m}N(t)\Bigl[\bar{Q}(1) - \frac{1}{m}\|\bar{Q}'\|_\infty\Bigr] \;>\; 0 \qquad \text{for every } m > \bar{m}.
\end{equation*}
The bound does not depend on the others' rates: the payoff of engine $i$ is strictly increasing in its own rate, and $e_i = 1$ is strictly dominant.
\end{proof}

Part (a) extends the myopic blindness of Theorem~\ref{thm:myopia}~(a) to any number of engines: below the opt-out threshold nothing in the instantaneous payoff pushes back, so no myopic profile settles there. Part (b) is stronger: past the critical number $\bar{m}$, full extraction is each engine's best move whatever the others play, and the corpus enters the extinction regime.

Long-run oriented competition converges to the erosion threshold from below, $\hat{e}(m) \to e^*$; myopic competition pushes extraction to its maximum, $\bar{E} = 1 > e^*$, as soon as there is enough engines.

\begin{remark}[scope of the game]\label{rem:game-scope}
The result is established in stationary rates under the average payoff criterion. The equal split of queries is a second restriction: with market shares $s_i$, the commons would respond to $\sum_i s_i e_i$, and what would matter is concentration rather than the number of engines. The symmetric case isolates the effect of the number alone.
\end{remark}

Theorem~\ref{thm:competition} constrains the survival mechanisms of Section~\ref{sec:mechanisms}: under competition, lowering one's rate alone is costly unless the other engines commit to the same. Outside the model, competitive pressure may also shorten the horizon over which engines optimize. If it does, Proposition~\ref{prop:myopic-comp} applies and extraction goes to its maximum. Either way, each engine extracts from a resource whose maintenance it does not pay for.

\subsection{Welfare}\label{subsec:welfare}

So far we have evaluated extraction only through the engine's average payoff. Users enter the model as a stream of queries and derive no surplus of their own, so the model cannot weigh their convenience gain from direct answers against the erosion of the commons. This subsection gives users a surplus, under the assumption most favorable to extraction: for a given corpus, users prefer more extraction.

\begin{assumption}[user surplus]\label{ass:user-surplus}
At corpus state $(N(t), Q(t))$, users derive the surplus
\begin{equation}
U(t) \;=\; \bigl[a\,e + b\,(1-e)\bigr]\, Q(t)\, N(t),
\qquad a > b \geq 0.
\end{equation}
\end{assumption}

The specification mirrors the engine's payoff of Section~\ref{subsec:myopia}. Queries draw on the informational value $Q(t) N(t)$ of the corpus. A fraction $e$ of this value reaches the user as direct answers, valued $a$ per unit, and a fraction $1-e$ as organic visits, valued $b$ per unit. The gap $a - b$ is the convenience premium of generative search. For a given corpus, $\partial U(t)/\partial e = (a-b)\,Q(t) N(t) > 0$, so more extraction benefits users. The premium fits the Delphic costs of \cite{broder2024}, {\em i.e.} the non-monetary costs of search: time, attention, interaction. The direct answer delivers the information at a lower Delphic cost than the visit it replaces. The case $b = 0$, where visits carry no value for users, is the most favorable to extraction and the assumption covers it.

Users in this model choose nothing. They take the answers as served, and they always prefer more extraction. Assumption~\ref{ass:user-surplus} makes them myopic on purpose. This myopia is individually rational, because one user's clicks leave the corpus unchanged, so giving up the direct answer costs convenience now and gains nothing later. The question is therefore where the users' interest lies once the corpus has responded to the rate.

Social welfare adds the users' surplus and the engine's payoff. As in Section~\ref{subsec:myopia}, we restrict our attention to stationary rates and evaluate the long run through the average payoff criterion. For $e < e^*$, $N(t) \to N^*(e)$ and $Q(t) \to \bar{Q}(e)$ (Proposition~\ref{prop:threshold}), so
\begin{equation}
W(e) \;:=\; \lim_{T \to \infty} \frac{1}{T}\int_0^T \bigl(U(t) + u(e,t)\bigr)\, dt
\;=\; \bigl[b + (1+a-b)\,e\bigr]\,\Phi(e),
\end{equation}
extended by $0$ on $[e^*, 1]$, where the corpus goes extinct. We call planner the hypothetical agent that would set the rate to maximize $W$. We denote by $e_{\mathrm{soc}}$ a maximizer of $W$. Theorem~\ref{thm:welfare}~(a) shows it is unique. Since $W(e) = (1+a-b)\,(\rho + e)\,\Phi(e)$ with
\begin{equation}
\rho \;:=\; \frac{b}{1+a-b} \;\geq\; 0,
\end{equation}
the maximizer depends on the preference parameters $(a,b)$ only through the ratio $\rho$.

The engine's payoff was written with the average quality of open publishers $\bar{Q}(e)$, and the users' surplus with the current quality $Q(t)$. The difference is only in the transition: along any trajectory with $e < e^*$, $Q(t) \to \bar{Q}(e)$, so both give the same long-run average.

The comparison below judges the planner and the engine by the same average payoff criterion. Under a discounted criterion, the planner would also value the transition, during which the volume is still high (Remark~\ref{rem:mutual-ruin}). We leave publisher surplus out on purpose. Adding it only lowers the social optimum (Remark~\ref{rem:publisher-surplus}), so the omission biases the planner toward extraction.

\begin{theorem}[the social optimum stays below the threshold]\label{thm:welfare}
Under Assumptions~\ref{ass:participation} to~\ref{ass:regularity}, the conditions of Proposition~\ref{prop:threshold}, and Assumption~\ref{ass:user-surplus}:
\begin{itemize}
\item[(a)] $W$ is continuous on $[0,1]$ and strictly concave on $[0, e^*]$; it attains its maximum at a unique $e_{\mathrm{soc}} \in [0, e^*)$: the social optimum lies strictly below the erosion threshold. Moreover $e_{\mathrm{soc}} > 0$ if and only if $\Phi(0) + \rho\,\Phi'_+(0) > 0$.
\item[(b)] $e_{\mathrm{soc}} \leq e_{\mathrm{eng}}$. If $b = 0$, then $e_{\mathrm{soc}} = e_{\mathrm{eng}}$. If $b > 0$ and $e_{\mathrm{eng}} \neq e_{\mathrm{out}}$, then $e_{\mathrm{soc}} < e_{\mathrm{eng}}$: as soon as organic visits retain any value for users and the sustainable optimum is not at the kink, the social optimum lies strictly below it.
\item[(c)] $e_{\mathrm{soc}}$ is nonincreasing in $\rho$, hence nonincreasing in $b$ and nondecreasing in $a$.
\item[(d)] Let $\bar{U}(e) := \lim_{T\to\infty}\frac{1}{T}\int_0^T U(t)\, dt = \bigl[b + (a-b)\,e\bigr]\,\Phi(e)$ denote the users' long-run surplus, and $e_{\mathrm{users}}$ its unique maximizer. Then
$e_{\mathrm{users}} \leq e_{\mathrm{soc}} \leq e_{\mathrm{eng}}$, and the three rates coincide when $b = 0$.
\end{itemize}
\end{theorem}

\begin{proof}
Throughout, write $A(e) := b + (1+a-b)\,e$. This affine function is strictly increasing, nonnegative on $[0,1]$, and $W = A\,\Phi$ on $[0, e^*]$.

(a) We first prove that $W$ is continuous on $[0,1]$. On $[0, e^*)$, $W = A\,\Phi$ is a product of continuous functions. At $e^*$, the left limit equals $A(e^*)\,\Phi(e^*) = 0$, which is the value of the extension. Beyond $e^*$, $W$ is constant.

We then prove that $W$ is strictly concave on $[0, e^*]$. Let $0 \leq x < y \leq e^*$, $t \in (0,1)$, and $z = t x + (1-t) y$. Concavity of $\Phi$ (Assumption~\ref{ass:regularity}), $A(z) > 0$ for $z > 0$, and the identity $A(z) = t A(x) + (1-t) A(y)$ give
\begin{align*}
W(z) \;=\; A(z)\,\Phi(z) \;&\geq\; A(z)\,\bigl[t\,\Phi(x) + (1-t)\,\Phi(y)\bigr]\\
&=\; t\,A(x)\Phi(x) + (1-t)\,A(y)\Phi(y) + t(1-t)\,\bigl[A(y)-A(x)\bigr]\,\bigl[\Phi(x)-\Phi(y)\bigr]\\
&>\; t\,W(x) + (1-t)\,W(y).
\end{align*}
The last inequality is strict because $A$ increases strictly while $\Phi$ decreases strictly on $[0, e^*]$.

A continuous and strictly concave function on the compact $[0, e^*]$ attains its maximum at a unique point, which is $e_{\mathrm{soc}}$. This point is not $e^*$, because $W(e^*) = 0$ while $W > 0$ on $(0, e^*)$. Finally, by strict concavity, the maximizer is $0$ exactly when $W'_+(0) \leq 0$, and $W'_+(0) = (1+a-b)\,\bigl[\Phi(0) + \rho\,\Phi'_+(0)\bigr]$.

(b) The argument of (a), applied to $A(e) = e$, shows that $J = e\,\Phi(e)$ is strictly concave on $[0, e^*]$. The maximizer $e_{\mathrm{eng}}$ is therefore unique, consistently with Theorem~\ref{thm:competition}, where $\hat{e}(1) = e_{\mathrm{eng}}$. By Theorem~\ref{thm:myopia}~(b), $e_{\mathrm{eng}} \in (0, e^*)$. Optimality at $e_{\mathrm{eng}}$ gives $J'_+(e_{\mathrm{eng}}) \leq 0$, that is $\Phi'_+(e_{\mathrm{eng}}) \leq -\,\Phi(e_{\mathrm{eng}})/e_{\mathrm{eng}}$. Hence
\begin{equation*}
W'_+(e_{\mathrm{eng}}) \;=\; (1+a-b)\,\Phi(e_{\mathrm{eng}}) + A(e_{\mathrm{eng}})\,\Phi'_+(e_{\mathrm{eng}})
\;\leq\; -\,\frac{b\,\Phi(e_{\mathrm{eng}})}{e_{\mathrm{eng}}} \;\leq\; 0.
\end{equation*}
If we had $e_{\mathrm{soc}} > e_{\mathrm{eng}}$, concavity of $W$ would give $W'_+(e_{\mathrm{eng}}) \geq [W(e_{\mathrm{soc}}) - W(e_{\mathrm{eng}})]/(e_{\mathrm{soc}} - e_{\mathrm{eng}}) > 0$, which contradicts $W'_+(e_{\mathrm{eng}}) \leq 0$. Therefore $e_{\mathrm{soc}} \leq e_{\mathrm{eng}}$. If $b = 0$, then $W = (1+a)\,J$ and the two maximizers coincide. If $b > 0$ and $e_{\mathrm{eng}} \neq e_{\mathrm{out}}$, then $\Phi$ is differentiable at $e_{\mathrm{eng}}$ (Assumption~\ref{ass:regularity}), so optimality gives the equality $\Phi'(e_{\mathrm{eng}}) = -\,\Phi(e_{\mathrm{eng}})/e_{\mathrm{eng}}$, and the computation above becomes $W'(e_{\mathrm{eng}}) = -\,b\,\Phi(e_{\mathrm{eng}})/e_{\mathrm{eng}} < 0$. The maximum of $W$ is therefore not at $e_{\mathrm{eng}}$, and since $e_{\mathrm{soc}} \leq e_{\mathrm{eng}}$, we conclude $e_{\mathrm{soc}} < e_{\mathrm{eng}}$.

(c) Let $\rho' > \rho \geq 0$, and let $e = e_{\mathrm{soc}}(\rho)$ and $e' = e_{\mathrm{soc}}(\rho')$ be the maximizers of $(\rho + \cdot)\,\Phi$ and $(\rho' + \cdot)\,\Phi$. Summing the two optimality inequalities $(\rho + e)\Phi(e) \geq (\rho + e')\Phi(e')$ and $(\rho' + e')\Phi(e') \geq (\rho' + e)\Phi(e)$ and simplifying gives $(\rho - \rho')\,[\Phi(e) - \Phi(e')] \geq 0$. Since $\rho < \rho'$, we get $\Phi(e) \leq \Phi(e')$, and $e \geq e'$ because $\Phi$ decreases strictly. Finally, $\partial \rho/\partial b = (1+a)/(1+a-b)^2 > 0$ and $\partial \rho/\partial a = -\,b/(1+a-b)^2 \leq 0$.

(d) The limit follows from Proposition~\ref{prop:threshold}, as for $W$. Since $a > b$, we can write $\bar{U}(e) = (a-b)\,(\rho_u + e)\,\Phi(e)$ with $\rho_u := b/(a-b) \geq 0$. The argument of (a), applied to the affine factor $b + (a-b)e$, gives strict concavity and a unique maximizer $e_{\mathrm{users}} = e_{\mathrm{soc}}(\rho_u)$. Since $1 + a - b > a - b$, we have $\rho_u \geq \rho$, and (c) gives $e_{\mathrm{users}} \leq e_{\mathrm{soc}}$. If $b = 0$, then $\rho_u = \rho = 0$, both $W$ and $\bar{U}$ are proportional to $J$, and the three maximizers coincide.
\end{proof}

Part (c) says which way the planner leans. The more users value the visit, the less extraction the planner tolerates. The more they value the direct answer, the more it tolerates.

In the $m$-engine game of Section~\ref{subsec:competition}, the commons responds to the average rate $\bar{E}$, and the users' surplus depends on it as well, so social welfare at a symmetric profile of rate $\hat{e}(m)$ is $W(\hat{e}(m))$.

\begin{corollary}[competition dissipates the users' surplus]\label{cor:dissipation}
Under the assumptions of Theorem~\ref{thm:welfare}:
\begin{itemize}
\item[(a)] for every $m \geq 1$, $e_{\mathrm{soc}} \leq e_{\mathrm{eng}} = \hat{e}(1) \leq \hat{e}(m) < e^*$: every number $m$ of engines, the single engine included, weakly overextracts relative to the social optimum;
\item[(b)] $W(\hat{e}(m)) \to 0$ as $m \to \infty$: competition dissipates not only the rent of the commons (Theorem~\ref{thm:competition}) but social welfare, the users' long-run surplus included;
\item[(c)] under the condition $e^* < e_{\mathrm{out}}$ of Theorem~\ref{thm:myopia} (c), the myopic rate satisfies $e_{\mathrm{myo}} > e^*$ and social welfare at $e_{\mathrm{myo}}$ is zero.
\end{itemize}
\end{corollary}

\begin{proof}
(a) combines Theorem~\ref{thm:welfare} (b) with Theorem~\ref{thm:competition} (b). 

(b): $\hat{e}(m) \to e^*$ by Theorem~\ref{thm:competition} (c), $W$ is continuous and $W(e^*) = 0$. (c) is Theorem~\ref{thm:myopia} (c) together with $W = 0$ on $[e^*, 1]$.
\end{proof}

Theorem~\ref{thm:welfare} (d) and Corollary~\ref{cor:dissipation} order every rate of the paper on one axis:
\begin{equation*}
e_{\mathrm{users}} \;\leq\; e_{\mathrm{soc}} \;\leq\; e_{\mathrm{eng}} = \hat{e}(1) \;\leq\; \hat{e}(m) \;<\; e^*,
\end{equation*}

with the myopic rate beyond the threshold of Theorem~\ref{thm:myopia} (c). 

The place of the users in this chain deserves an explanation: they prefer more extraction at every query, yet the rate that serves their long-run interest is the lowest of the three. Raising the rate benefits every agent through the direct answers and costs every agent through the erosion of the commons. Users also lose the visits, which still carry a value $b$ for them, so their optimum is the lowest. The engine loses nothing through visits, so its optimum is the highest. The planner counts the users' loss, diluted with the engine's payoff, and its optimum lies between the two.

The first three rates of the ordering are optima: each says where an agent's long-run interest lies. Only the equilibrium rates $\hat{e}(m)$ describe a behavior. This distinction carries a practical warning about measurement. By construction, the surplus $U(t)$ increases with $e$ for a given corpus. An engine or a regulator that adjusts the extraction rate to this measured surplus therefore pushes the system past $e^*$, because the signal remains positive at every corpus and never shows the threshold where the users' long-run surplus $\bar{U}$ reaches zero.

Part (b) has a simple cause: the engine ignores the user value $b\,(1-e)\,\Phi(e)$ that its extraction destroys, so even a long-run oriented single engine overextracts. The unit weight on the engine's payoff in $W$ restricts nothing. Take any welfare function that weights answers, visits and engine payoff with nonnegative coefficients, and suppose that more extraction still raises it for a given corpus, as in Assumption~\ref{ass:user-surplus}. Up to a positive factor, this function is again of the form $(\rho + e)\,\Phi(e)$ for some $\rho \geq 0$, and Theorem~\ref{thm:welfare} applies to it.

\begin{remark}\label{rem:corner}
By Theorem~\ref{thm:welfare} (a), the planner shuts extraction down entirely when $\rho \geq \Phi(0)/\lvert\Phi'_+(0)\rvert$ (a finite bound since $\Phi'_+(0) < 0$). Below the bound, the social optimum is interior.
\end{remark}

\begin{remark}[publisher surplus]\label{rem:publisher-surplus}
Adding publishers reinforces the theorem. At rate $e$, a publisher of quality $\theta$ earns $\max\bigl(r(1-e)\theta,\; \omega\theta - \kappa\bigr)$, as in Section~\ref{sec:dilution}. This revenue never increases with $e$. Neither does the aggregate publisher surplus. Adding that surplus to $W$ moves every maximizer weakly below $e_{\mathrm{soc}}$. Every market structure still overextracts, and by a wider margin. A large enough publisher surplus pushes the optimum below $e_{\mathrm{users}}$.
\end{remark}

\section{Survival mechanisms of the ecosystem}\label{sec:mechanisms}

A mechanism is a solution to our problem if it moves the system away from the extinction regime, either by lowering the extraction rate $e$, or by shifting the threshold $e^*$ upward through $n$ or $\lambda$, or by delaying opt-out and preserving corpus quality through $e_{\mathrm{out}}$. 

We now present seven mechanisms, grouped according to their coordination prerequisites: what can be done by the engine alone, what requires a market, and what requires collective coordination. 

\subsection{Architecture mechanisms: implementable by the engine alone}

\paragraph{M1. The referral floor.} The engine imposes on itself a floor of outbound traffic to the sources: prominent citations, actionable links in the answer, a published and quantified commitment on the ratio between pages crawled and visits returned. In the model, M1 directly lowers $e$. It is the simplest and most robust mechanism: it requires no agreement from anyone and is directly measurable. Since all three channels depend on $e$, it acts on them together. Restitution is a design variable: the gap between pages crawled and sources credited already varies strongly from one GSE to another \cite{strauss2026}.

\paragraph{M2. The teaser answer.} The engine voluntarily bounds what it extracts from each source: the answer addresses the question, but the full data remain behind the click. M2 lowers $e$ , hence reducing depreciation through $\lambda(e)$: if citation becomes profitable again, the shift toward perishable content stops.

\subsection{Market mechanisms: implementable by contract}

\paragraph{M3. Paid crawling.} Access to content is paid per crawl request, at a price set by the publisher. Its primary function is compensation: the payment improves the option of staying accessible, which raises $\bar{\theta}$ and keeps the best publishers open longer. Its filtering function, by contrast, is conditional and we present it as such: a price is not in itself a quality signal. The price sorts quality only if the GSE holds a credible estimate of the marginal value of the content and refuses to pay above it. 

\paragraph{M4. Flow-indexed licensing.} The engine pays for new production, not for the archive: the license price is proportional to the publisher's flow of new content. With M4, the engine pays for exactly what it needs to survive: freshness. The license ties its extraction cost back to the health of the corpus. In the model, M4 adds license revenue to the funded flow, raising the renewal coefficient $n$ (equivalently, a higher effective $\pi$ in Assumption~\ref{ass:renewal}), and targets the publishers whose contribution to $n$ is largest. 

\paragraph{M5. Attribution-based compensation.} Every generated answer distributes a micro-payment across its sources, in proportion to their effective contribution to the answer. The reference method is a Shapley-value attribution \cite{shapley1953}: each source is paid its average marginal contribution to the answer. The method has been applied to retrieved documents \cite{ye2025, nematov2025}, with approximations adapted to generative search \cite{maxshapley2025}. 

\subsection{Coordination mechanisms: implementable by publishers or the regulator}

\paragraph{M6. Collective bargaining.} A lone publisher's marginal contribution to the corpus is small, so its threat to opt out is weak. The European press publishers' right \cite{cdsm2019} subjects the online use of press publications by online services to the publishers' authorization. A collective management structure on this model aggregates the long tail into a single bargaining party, whose collective opt-out would degrade the engine's quality. It provides a compensation path for small publishers without solving the attribution problem of M5, and a counterweight that can make M3, M4 and M5 negotiable on real terms.

\paragraph{M7. The enforceable usage signal.} M7 extends robots.txt into a machine-readable declaration of content usage terms: crawling allowed in exchange for citation, for payment, or forbidden, with legal force. The technical standard exists in draft form in several initiatives. What is missing is enforceability. This mechanism goes beyond the model. In the binary open-or-closed choice, reducing the opt-out cost $\kappa$ hastens opting out and degrades the crawlable corpus. The real function of M7 is to create a third state absent from the model, conditional access, with a contracting cost distinct from $\kappa$ that the mechanism would drive toward zero. The publisher then states its price instead of choosing between free openness and opting out. Formalizing this third state would be an extension of the model.

\subsection{Combining the mechanisms}

No mechanism is sufficient alone, but they differ in their prerequisites and in how soon they can start. M1 and M2 require nobody's agreement and can start immediately. M3, M4 and M5 require a market, whose first building blocks exist. M6 and M7 require collective or regulatory coordination, the slowest to build. 

The welfare ordering gives three further implications for how these mechanisms fit together. First, bringing a competitive extraction rate $\hat{e}(m)$ down toward $e_{\mathrm{soc}}$ through M1 or M2 can benefit users themselves, even though they prefer more extraction for a given corpus. Second, whenever 
$\hat{e}(m) > e_{\mathrm{eng}}$, reducing extraction also benefits the engines collectively. In that range, M1 and M2 raise the aggregate payoff and preserve the corpus at the same time. Third, the compensation mechanisms M3, M4 and M5 cannot by themselves solve the competitive externality. By supporting participation or renewal, they can preserve the corpus and shift the relevant thresholds, but each engine still has an incentive to extract against the value preserved by the others. Compensation and lowering the rate are therefore complements when several engines share the commons.

M6 and M7 play a different role. The implications above compare extraction rates, and neither M6 nor M7 sets a rate. Both act on the conditions under which the other mechanisms operate: M6 gives publishers the bargaining power that makes M3, M4 and M5 negotiable on real terms, and M7 adds an option, conditional access, that the binary open-or-closed choice does not contain.

Table~\ref{tab:mechanisms} summarizes the correspondence between each mechanism and the model parameter it repairs.

\begin{table}[h]
\centering
\small
\begin{tabular}{llll}
\toprule
Mechanism & Implemented by & Acts on & Prerequisite \\
\midrule
M1 Referral floor & Engine & $e \downarrow$ & None \\
M2 Teaser answer & Engine & $e \downarrow$ & None \\
M3 Paid crawling & Engine + publisher & $\bar{\theta} \uparrow$, $\bar{Q} \uparrow$ & Market \\
M4 Flow-indexed licensing & Engine + publisher & $\pi \uparrow$, $n \uparrow$ & Market \\
M5 Attribution  & Engine + publishers & revenue $\uparrow$  & Market \\
M6 Collective bargaining & Publishers & market power & Coordination \\
M7 Enforceable signal & Regulator & extension: conditional access & Coordination \\
\bottomrule
\end{tabular}
\caption{The seven mechanisms and the model parameter each acts on.}\label{tab:mechanisms}
\end{table}

\section{Conclusion}\label{sec:conclusion}

The crawlable corpus is a renewable common-pool resource whose stock depends on the extraction imposed on it, not a fixed input to generative search. We established that GSEs degrade its three quantities (volume, quality and lifetime) at once. Below the erosion threshold the degradation is continuous. Above it the corpus goes extinct.

A myopic GSE can cross the threshold, a long-run oriented GSE stays below it. Competition worsens the phenomenon: with several engines, the symmetric equilibrium extraction rate converges to the threshold. Moreover, under the assumption most favorable to extraction, that users strictly prefer direct answers, the socially optimal extraction rate lies strictly below the threshold.

Avoiding this collapse requires mechanisms that either lower the extraction rate or fund the renewal of the corpus directly. A generative search ecosystem is viable only if its extraction rate leaves enough revenue to pay for that renewal.

\newpage
\appendix
\section{Numerical instantiations}\label{sec:appendix}

This appendix instantiates the model on an explicit parameter set, with two goals: to prove that Assumption~\ref{ass:regularity} is satisfiable, as announced in Section~\ref{subsec:competition}, and to prove that the kink plateau of Theorem~\ref{thm:competition}~(b) does occur. Take $f$ uniform on $[0,1]$, the simplest density, and the parameter set $r = 1$, $\omega = 0.6$, $\kappa = 0.3$, $\pi = 1.2$, $\lambda(0) = 0.05$, $\delta = 0.1$, $K = 1$. For the uniform density, $\mathbb{E}[\theta] = 1/2$, $\mu(e) = \bar{\theta}(e)^2/2$ and $\bar{Q}(e) = \bar{\theta}(e)/2$. Then $e_{\mathrm{out}} = 0.70$ and $e^* \approx 0.738$. The parameters are chosen for this: the kink is interior to $[0, e^*]$, so opting out begins before extinction and all three channels are active below the threshold. Figure~\ref{fig:appendix} plots the resulting steady state and equilibria.

\paragraph{Satisfiability of Assumption~\ref{ass:regularity}.} Concavity holds analytically, and three facts establish it: $\Phi'' < 0$ below the kink, $\Phi'' < 0$ above it, and a slope that drops at the kink. Below the kink, $\Phi(e) = \frac{1}{2} - \frac{1+2e}{24(1-e)}$, so $\Phi''(e) = -\frac{1}{4(1-e)^3} < 0$. Above the kink, $\Phi(e) = \frac{50e^3 - 15e^2 + 15e - 23}{36(5e^2 - 7e + 2)}$ and $\Phi''(e) = \frac{200e^3 - 375e^2 + 285e - 83}{2(e-1)^3(5e-2)^3}$. On $[e_{\mathrm{out}}, e^*]$, the denominator is negative. The numerator is positive there. Its derivative $600e^2 - 750e + 285$ has negative discriminant and positive leading coefficient, hence is positive everywhere, so the numerator increases, and it already equals $27/20$ at $e_{\mathrm{out}}$. Hence $\Phi'' < 0$ above the kink as well. At the kink itself, the slope drops from $\Phi'_-(e_{\mathrm{out}}) = -25/18$ to $\Phi'_+(e_{\mathrm{out}}) = -25/6$. The function $\Phi$ is therefore concave on all of $[0, e^*]$, and Assumption~\ref{ass:regularity} holds.

\paragraph{The kink plateau.} The right panel of Figure~\ref{fig:appendix} shows the trajectory predicted by Theorem~\ref{thm:competition}~(b): as the number of engines grows, the equilibrium rises from the single-engine optimum, pauses at the kink, and resumes toward the threshold. In numbers, $e_{\mathrm{eng}} \approx 0.537$, $\hat{e}(2) \approx 0.615$, $\hat{e}(5) \approx 0.690$, then $\hat{e}(m) = e_{\mathrm{out}} = 0.70$ for $m = 6$ to $17$, then $\hat{e}(20) \approx 0.704$ and $\hat{e}(50) \approx 0.724$, against $e^* \approx 0.738$. The kink plateau therefore occurs in this instantiation, and it is finite: twelve values of $m$, after which strict growth resumes.
\begin{figure}[h]
\centering
\includegraphics[width=0.98\textwidth]{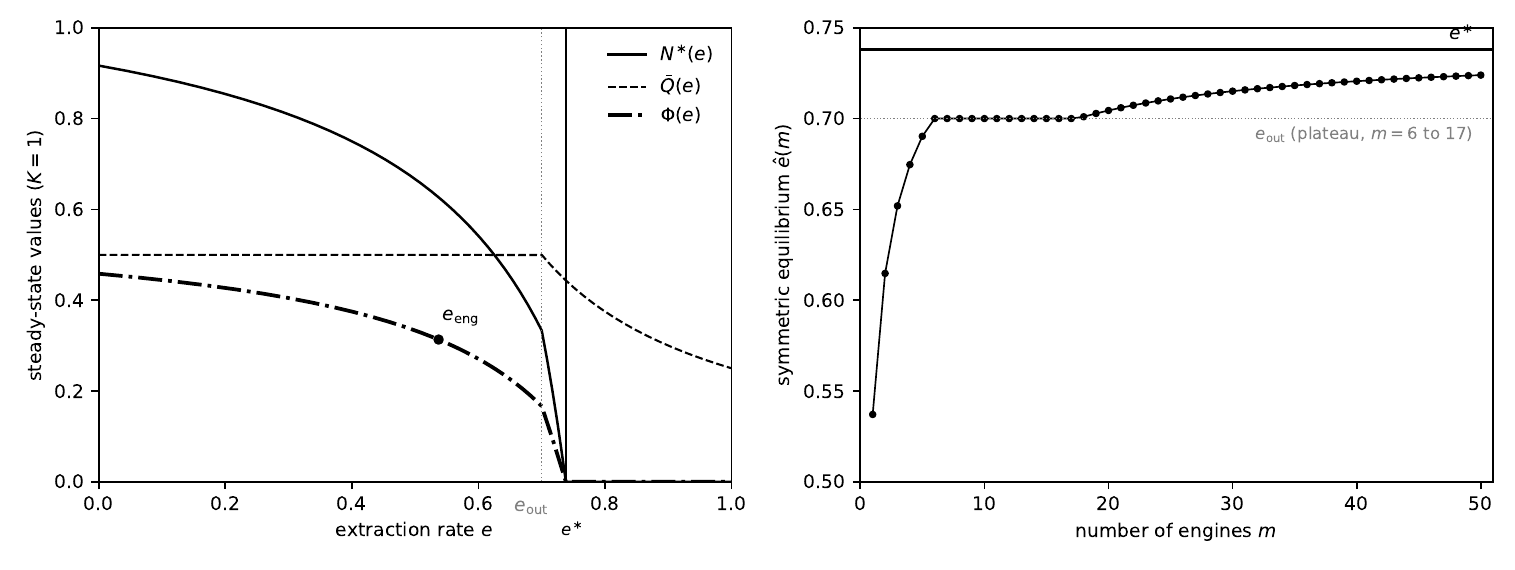}
\caption{Left: steady-state volume $N^*(e)$, average quality of open publishers $\bar{Q}(e)$, steady-state value of the commons $\Phi(e)$, with the sustainable optimum $e_{\mathrm{eng}}$, the kink $e_{\mathrm{out}}$ and the erosion threshold $e^*$. Right: the sustainable symmetric equilibrium $\hat{e}(m)$ as a function of the number of engines, with the kink plateau ($m = 6$ to $17$) and convergence to $e^*$.}\label{fig:appendix}
\end{figure}

\end{document}